\PassOptionsToPackage{dvipsnames,x11names}{xcolor}
\documentclass[sigconf]{acmart}

\setcopyright{none}
\renewcommand\footnotetextcopyrightpermission[1]{}

\usepackage[capitalize,nameinlink]{cleveref}
\graphicspath{{figures/}}
\usepackage{macros}

\begin{document}

\title{A Re-Keying Game for Encrypted Multi-Agent Control under Persistent Threats}

\author{Sai Sandeep Damera}
\orcid{0000-0003-3419-9719}
\email{sdamera@umd.edu}
\affiliation{%
  \institution{University of Maryland}
  \city{College Park}
  \state{Maryland}
  \country{USA}}

\author{John S. Baras}
\orcid{0000-0002-4955-8561}
\email{baras@umd.edu}
\affiliation{%
  \institution{University of Maryland}
  \city{College Park}
  \state{Maryland}
  \country{USA}}

\begin{abstract}
Fully homomorphic encryption (FHE) lets a coordinator control a fleet on encrypted states, but CKKS, the approximate scheme for real-valued arithmetic, returns noisy decryptions from which an adversary can recover the secret key. Flooding each decryption with noise defends it, but a fleet must decrypt to act, so the count of decryptions grows every cycle while the loop tolerates only bounded noise. We formulate the choice of flooding and key lifetime as a leader-follower game between the fleet operator and an advanced persistent threat that listens first and then spoofs the agents' sensors. We show that the closed loop decides how long a key may be used: the formation's specification caps each key's decryption budget through the loop's $H_2$ norm, and only re-keying, not bootstrapping, renews it. A spoof consumes part of the specification that flooding needs, and the coordination graph and the cost of an alarm decide in closed form whether the scan alone, a lower flood or nothing deters it. On $196$ agents escorting a convoy at 128-bit security, five graphs with the same spectral radius differ $18$-fold in key lifetime, a lower flood deters the spoofer on the torus for a twentieth of a key a day, and on the ring it cuts the damage $7$ times. More broadly, any system that must repeatedly decrypt to act has a key lifetime set by its tolerance to noise, making key management part of control design.
\end{abstract}


\maketitle
\pagestyle{plain}

\section{Introduction}
\label{apte:sec:intro}

\begin{figure}[t]
\centering
\includegraphics[width=\columnwidth]{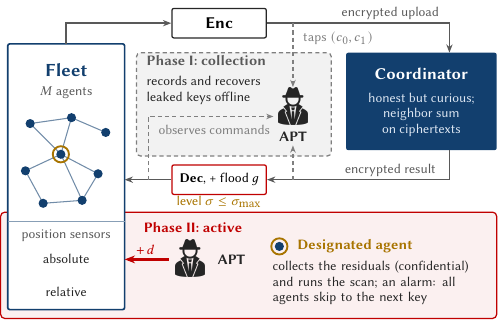}
\caption{The encrypted loop and the adversary's two phases.}
\label{apte:fig:threat}
\end{figure}

Autonomous fleets increasingly offload coordination to the cloud, and when the agents' states and commands are private, fully homomorphic encryption (FHE) lets the cloud compute the coordination on encrypted data, so that neither the server nor an eavesdropper sees plaintext~\cite{kim2016encrypting,9438025}, now also for multi-agent fleets over an unbounded horizon~\cite{damera2026endtoendencryptedcontrolpipeline}.

Confidentiality of the computation is not the end of the story. CKKS, the approximate scheme for real-valued arithmetic, returns a decryption that still carries the encryption noise~\cite{cheon2017homomorphic}, and an adversary who observes decryptions can reconstruct the secret key~\cite{li2021security,guo2024key}. In an encrypted fleet the decryptions are unavoidable, since the actuators must decrypt to move, so the leak is intrinsic to operating the loop. The security of approximate FHE treats the flooding of a single decryption as a static property~\cite{li2022securing,bergamaschi2025revisiting}, encrypted control assumes an honest-but-curious cloud~\cite{9438025}, and game-theoretic models of persistent threats never reach inside the cryptosystem~\cite{van2013flipit,pawlick2015flip}. A fleet facing a real, adaptive adversary sits in the intersection of all three.

We model that adversary as an advanced persistent threat (APT) whose aim is the fleet's data, acting in two phases (\cref{apte:fig:threat}) as in game-theoretic APT models for robotic autonomy that separate a stealthy penetration phase from an active damage phase~\cite{zoulkarni2024defending,zoulkarni2025probabilistic}. Throughout, it records the encrypted traffic and the released decryptions, so that it can recover offline any key whose epoch leaked enough. The coordinator follows the protocol on an authenticated channel, so once active the APT can act only on the agents' position sensors, which measure the physical world in the clear. A spoof of those sensors never touches a ciphertext, yet its formation error uses part of the specification that flooding needs, so holding the formation against it costs the defender key life. We cast the defense as a leader-follower game in which the operator commits its flooding and its response to an alarm, and the APT, which reads the flooding, chooses how hard to spoof on each segment of the fleet's route.

The defender's adaptive move is re-keying, switching every agent to a fresh key, not the bootstrapping that refreshes the noise under the same key and leaves the adversary's accumulated knowledge intact. Every unit of flooding the defender gives up to the spoofer shortens each key's budget, so the re-key rate becomes the cost of physical defense, and whether re-keying is needed at all turns on the security target and the length of the mission.

We make the following contributions.
\begin{itemize}
 \item We pose the security of an encrypted fleet as a leader-follower game against an APT that collects decryptions and spoofs the agents' position sensors (\cref{apte:sec:game}).
 \item We show that the loop's $H_2$ tolerance to flooding sets each key's decryption budget, which only re-keying resets, and give the cost of security in closed form (\cref{apte:sec:passive,apte:sec:active}).
 \item We characterize the equilibrium over both phases in closed form: the coordination graph and the APT's cost of an alarm decide whether the scan alone, a lower flood or nothing deters the spoofer (\cref{apte:sec:active}).
 \item On $196$ agents at 128-bit security, five graphs of equal spectral radius differ $18$-fold in key lifetime, and the game's payoff matrices show all three regimes (\cref{apte:sec:evaluation}).
\end{itemize}

\section{Related Work}
\label{apte:sec:related}

Three threads border this work: encrypted control, the security of approximate FHE, and game-theoretic models of persistent threats; unlike any of them we let the control loop set the re-keying schedule through which the defender pays for its defense.

\paragraph{Encrypted control.}
Encrypted control moved from networked controllers to fully homomorphic ones~\cite{kim2016encrypting}, with tutorial treatments consolidating the area~\cite{9438025}. More recent work sustains an unbounded operating horizon through bootstrapping with stability guarantees~\cite{schlor2024bootstrapping} and scales to an end-to-end multi-agent CKKS fleet pipeline~\cite{damera2026endtoendencryptedcontrolpipeline}. This line protects the loop's confidentiality against an honest-but-curious cloud, and verifiable homomorphic encryption adds integrity when the cloud is not trusted~\cite{chatel2024veritas}. We keep that coordinator and ask how long its key survives an outside adversary that exploits the unavoidable decryptions.

\paragraph{Security of approximate FHE.}
Li and Micciancio showed that CKKS decryptions leak the secret key and introduced the IND-CPA\textsuperscript{D} notion to capture it~\cite{li2021security}, with a noise-flooding countermeasure calibrated through differential privacy~\cite{li2022securing}. Later work recovered keys under non-worst-case flooding from very few decryptions on a production library~\cite{guo2024key} and mapped the concrete security of flooding against a bounded query budget~\cite{bergamaschi2025revisiting}. This thread fixes the security of a single decryption or a fixed query count; we take its leakage as a primitive and schedule re-keying against leakage that accumulates over the horizon.

\paragraph{Game-theoretic models of persistent threats.}
FlipIt cast the contest for a stealthily compromised resource as a timing game in which a defender periodically reclaims control~\cite{van2013flipit}, later specialized to cloud advanced persistent threats~\cite{pawlick2015flip} and to the timing of key resets~\cite{zheng2017reset}, but in this line the resource is abstract and none reaches inside the cryptosystem. Closest to our formulation are multi-phase game-theoretic models of APTs against robotic autonomy, which split an attack into a stealthy penetration phase and an active damage phase acting on the autonomy stack~\cite{zoulkarni2024defending,zoulkarni2025probabilistic}. We inherit that decomposition and carry it inside the cryptosystem: the first phase collects decryptions against the secret key, the second acts on the agents' sensors, and the two meet in the flooding that both the key's budget and the formation's specification draw on.

\section{System Model and Preliminaries}
\label{apte:sec:setting}

A fleet of agents holds a formation through a coordinator that never sees their state. Each agent encrypts what it uploads, the coordinator computes each agent's share of the coordination law on the ciphertexts, and each agent decrypts only its own result.

\paragraph{Fleet and cooperative control.}
The fleet has $M$ agents indexed by $\mathcal{V}=\{1,\dots,M\}$, each a fully actuated planar double integrator sampled at period $dt$~\cite{damera2026endtoendencryptedcontrolpipeline}, with state $x_i(k)=(p_i(k),v_i(k))\in\R^{n}$, $n=4$,
\begin{equation}
\label{apte:eq:agent}
x_i(k+1) = A\,x_i(k) + u_i(k),
\qquad A=I_2\otimes\begin{pmatrix}1 & dt\\ 0 & 1\end{pmatrix},\ u_i(k)\in\R^{n}.
\end{equation}
The agents coordinate over an undirected, connected graph $\mathcal{G}=(\mathcal{V},\mathcal{E})$ with neighbor sets $\mathcal{N}_i$, degrees $d_i$ and Laplacian $L$; we consider the chain, the ring, the torus, the hub-and-spoke and the complete graph $K_M$. Displacement-based feedback on $(x_i-x_j)-(r_i-r_j)$, with an attraction of gain $\gamma>0$ toward each agent's offset in the formation $r\in\R^{Mn}$, gives
\begin{equation}
\label{apte:eq:consensus}
x(k+1)-r = \bigl((1-\gamma)\,I_{Mn}-\varepsilon\,(L\otimes I_n)\bigr)\bigl(x(k)-r\bigr)
\end{equation}
with step size $\varepsilon>0$~\cite{olfati2004consensus}. With $0=\mu_1<\mu_2\le\dots\le\mu_M$ the eigenvalues of $L$, the closed-loop spectral radius is $\rho_{\mathrm{cl}}=\max_i\lvert 1-\gamma-\varepsilon\mu_i\rvert$, and outside a band of width $\gamma$ below instability
\begin{equation}
\label{apte:eq:rhoconst}
\rho_{\mathrm{cl}} = 1-\gamma
\end{equation}
for every topology, a constant of the control law attained on the consensus direction, which the formation metric removes with the centroid. The loop's response to injected noise, by contrast, depends on the whole spectrum~\cite{bamieh2012coherence,patterson2014consensus}, and that is where the topology enters.

\paragraph{Encrypted computation.}
CKKS~\cite{cheon2017homomorphic} encrypts a vector of $N/2$ real numbers in one ciphertext, $8192$ at the ring dimension $N=16384$ we use; ciphertexts add and multiply with plaintext vectors entry by entry, and a rotation shifts the vector cyclically, which is how a value moves from one agent's entries to another's. Encrypting $z$ gives $\mathsf{ct}=(c_0,c_1)$ over $\Z_q[X]/(X^{N}+1)$ such that, with ternary secret key $s$ and scale $\Delta$,
\begin{equation}
\label{apte:eq:ckks}
c_0 + c_1\,s = \Delta\,\mathsf{ecd}(z) + \eta \pmod{q},
\end{equation}
with $\eta$ a small error, so decryption returns $\tilde z=z+e/\Delta$ with $e:=\mathsf{dcd}(\eta)$, never $z$ exactly. Every operation grows the error, and \emph{bootstrapping} refreshes a ciphertext under the \emph{same} key $s$, which lets an encrypted controller run indefinitely~\cite{schlor2024bootstrapping}; it changes the noise, not the key.

\paragraph{The encrypted fleet-control loop.}
In the pipeline of~\cite{damera2026endtoendencryptedcontrolpipeline}, with the observer moved to the agents, each agent uploads the offset $\hat x_i(k)-r_i$ of its estimate encrypted under the fleet's public key, the coordinator computes the neighbor sum $\varsigma_i(k)=\varepsilon\sum_{j\in\mathcal{N}_i}(\hat x_j(k)-r_j)$ with plaintext weights and rotations at multiplicative depth one, with no bootstrap and nothing retained between cycles, and each agent decrypts its result and releases its target
\begin{equation}
\label{apte:eq:target}
t_i(k) = (1-\gamma-\varepsilon d_i)\bigl(\hat x_i(k)-r_i\bigr) + r_i + \varsigma_i(k) + g_i(k),
\end{equation}
with $g_i(k)\sim\mathcal{N}(0,\sigma^2 I_n)$ the flooding it draws itself, and applies $u_i(k)=t_i(k)-A\,\hat x_i(k)$. Under (A4) the decrypted value is the honest neighbor sum up to an error that \cref{apte:lem:honest} bounds by $B$. The loop releases $Mn$ values per cycle, so the count of releases under a key grows at a fixed rate whatever the loop's contraction.

\paragraph{Standing assumptions.}
(A1) Contractive loop: $\gamma>0$ and $\varepsilon\mu_M<2-\gamma$. (A2) Single-key pipeline with a key book. Every agent encrypts and decrypts under one key $s$ shared by the fleet, drawn in turn from a book of keys that every agent holds from before the mission, and the agents switch keys in lockstep on a public schedule. (A3) Bounded states: each agent clips its uploads to a known range $[-w,w]$. (A4) Honest-but-curious coordinator. The coordinator computes the neighbor sum as specified and learns nothing beyond the ciphertexts it is sent, and every ciphertext travels between the agents and the coordinator over an authenticated channel, so the adversary reads the encrypted traffic and does not alter it. (A5) Relative sensing. Besides its own position, each agent measures the relative position of each of its graph neighbors. (A6) Alarm channel. One agent of the fleet, the designated agent, collects every agent's relative-sensing residuals over confidential, authenticated links and announces an alarm to every agent on them. A residual carries sensing noise, estimation error and the signature of any spoof, not a position.

\section{Threat Model and the Game}
\label{apte:sec:game}

\paragraph{Adversary.}
We model the adversary as an advanced persistent threat (APT) whose aim is the fleet's data. It records the encrypted traffic and the released decryptions, recovers offline any key whose epoch leaked enough, and then decrypts that epoch's ciphertexts at leisure, so a broken key exposes only the traffic sent under it. It also observes the agents physically and learns the commands they apply. From a tapped ciphertext and a known plaintext entry $z$, \cref{apte:eq:ckks} gives
\begin{equation}
\label{apte:eq:leak}
c_1\,s = \Delta\,z + e - c_0 \pmod{q},
\end{equation}
a noisy linear constraint on $s$~\cite{li2021security}. With full plaintext knowledge this is the established known-plaintext key recovery, which succeeds from very few decryptions and is defeated only by flooding the released values~\cite{guo2024key,li2022securing}.

The APT acts in two phases. In the \emph{collection} phase it only records, and nothing the agents measure changes. In the \emph{active} phase it also acts on the one channel that (A4) leaves open, adding a bias $d_i(k)$ to the absolute position that agent $i$ measures on a contested stretch of the route. The relative sensing of (A5) stays out of its reach. Each agent's observer trusts the biased measurement, so its estimate, its upload and the fleet's coordination follow it, and for every observer gain the formation error, measured with the centroid removed, moves by $-Pd$ with $P=I-\mathbf{1}\mathbf{1}^\top/M$. The spoof never touches a ciphertext. Its formation error uses part of the specification that flooding needs, so the defender must either flood less, which shortens every key's budget, or let the formation leave its specification.

\paragraph{Detection.}
The defender sees a spoof only through the relative sensing. Each agent compares the relative positions it measures with those that its estimate and the decrypted neighbor sum predict. The vector of these residuals has mean $Ld$, so a pattern along the graph's slow modes is hard to see and a common shift of a connected group is invisible. The designated agent of (A6) collects the residuals, whitens them in time and runs a scan of energy tests on aligned blocks of $1$ to $3000$ cycles. It puts $0.9$ of its false-alarm level $10^{-3}$ on the longest block and spreads the rest over the eleven shorter scales, so that neither a short burst nor a slow drift hides at a scale the scan leaves out. Each agent's innovation gate, a standard check of navigation filters, bounds the energy of a single cycle. A block whose energy crosses its threshold raises an alarm, and the designated agent announces it to the fleet. The size $\lambda$ of a spoof is the energy it adds to the whitened residuals of a window, which is the noncentrality of that window's energy test, divided by $\Lambda$, the added energy at which a single window's test at level $0.05$ raises an alarm half the time.

\paragraph{The game.}
The defender, the fleet's operator, leads. Before the mission it commits, for each segment $s$ of the contested stretch, a flood level $\sigma_s\in(0,\sigma_{\max,s}]$ at which the agents there flood from the first lap. Flooding below $\sigma_{\max,s}$ leaves part of the specification free for a spoof's error and shortens every key serving the segment. It also commits the scan and its response to an alarm: one off-schedule re-key, in which every agent skips to the next key of its book, retiring everything the APT collected under the old key. After an alarm the active campaign ends with probability $p$, for instance because the agents under attack switch to an independent backup position sensor. From the known plaintexts of (T1) the APT reads every flood sample to within $B$, and it detects a change of the flood level within a cycle on the ring and within about $33$ cycles on the torus, so the commitment is observed, as the leader-follower model requires. The APT follows. In each window $w$ of each segment it chooses the size $\lambda_{s,w}\ge0$ of its spoof, with $\lambda_{s,w}=0$ when it stays dormant. A flood profile $\sigma=(\sigma_s)_s$ and a spoof profile $\lambda=(\lambda_{s,w})_{s,w}$ form a strategy profile.

Once active, the APT gains the formation's expected excess over its specification,
\begin{equation*}
X(\sigma,\lambda)=\mathbb{E}\sum_{s,w}\bigl(e_{s,w}(\sigma,\lambda)-\rho_{\mathrm{spec},s}^2\bigr)^{+},
\end{equation*}
where $e_{s,w}$ is the formation's mean square over window $w$ of segment $s$, the window of an alarm counts up to the cycle of the alarm, and nothing counts after the campaign ends. It pays for every alarm it sets off, so its payoff is
\begin{equation}
\label{apte:eq:payoffA}
U_A(\sigma,\lambda) = X(\sigma,\lambda) - c_{\mathrm{det}}\,P(\sigma,\lambda),
\end{equation}
where $P(\sigma,\lambda)$ is the probability that the scan raises an alarm during the campaign and $c_{\mathrm{det}}>0$ is the APT's cost of an alarm, which re-keys the fleet and ends its campaign with probability $p$. The defender pays for the keys it consumes and for the excess,
\begin{equation}
\label{apte:eq:payoffD}
J_D(\sigma,\lambda) = c_{\mathrm{key}}\bigl(K(\sigma)+P(\sigma,\lambda)\bigr) + w_{\mathrm{ex}}\,X(\sigma,\lambda),
\end{equation}
where $K(\sigma)$ counts the keys of a mission in which every key keeps the budget \cref{apte:eq:budget} of the segments it serves, $P(\sigma,\lambda)$ counts the expected re-key at an alarm, and $r=c_{\mathrm{key}}/w_{\mathrm{ex}}$ is the price of a key in units of formation excess. A Stackelberg equilibrium is a commitment $\sigma^\star$ that minimizes $J_D(\sigma,\lambda^\star(\sigma))$, where $\lambda^\star(\sigma)$ maximizes $U_A(\sigma,\cdot)$. The APT's aim on the data does not enter $U_A$, because the decryption budget fixes it. Every key keeps its budget, so the chance that the APT breaks any key stays below $2^{-\nu}$ in both phases, whatever it does. What the spoof changes is the cost of that guarantee, through $K$, and this coupling is how the two phases meet in one game.

Four assumptions bound the model.
(T1) Known plaintext from physical observation is the worst case for the defender.
(T2) The pipeline is single-key, so there is one key to protect and one release channel.
(T3) The adversary is computationally bounded: CKKS is $(c+\log_2 24)$-bit IND-CPA secure, and recovering $s$ from the public key, the evaluation keys and encryptions of known plaintexts is hard.
(T4) Corrupted agents as strategic players are left to future work.

\section{Phase 1: The Passive Game}
\label{apte:sec:passive}

In the passive phase the adversary only listens, so the defender commits its flooding level $\sigma$ in advance against an observation budget $\bar Q$, which for a persistent adversary is the whole mission of $\bar Q$ cycles. The interaction is a one-shot Stackelberg game governed by the noise-flooding theorem of Li et al.\ for approximate FHE~\cite{li2022securing}.

\paragraph{The decryption budget.}
Flooding adds noise to the value released by \cref{apte:eq:leak}, so resolving the key takes more observations~\cite{li2021security,li2022securing,bergamaschi2025revisiting}.
Under (A4) every released value lies within $B$ of the honest neighbor sum, a value computable without $s$ (\cref{apte:lem:honest}), so the agents' Gaussian flooding changes the distribution of the $n_{\mathrm{rel}}=Mn$ values of a cycle by a Kullback-Leibler divergence of at most $n_{\mathrm{rel}}B^2/(2\sigma^2)$~\cite[Lemma~6]{li2022securing}, which adds over cycles.
Theorem~5 and Corollary~2 of~\cite{li2022securing} turn this into $(c,\nu)$-bit IND-CPA$^{\mathrm D}$ security, with $\nu$ the statistical security of one key's releases, for
\begin{equation}
\label{apte:eq:budget}
Q(\sigma) = \frac{\sigma^2}{c_f\,2^{\nu}B^2},
\qquad c_f = 24\,n_{\mathrm{rel}},
\end{equation}
cycles under (T3) (\cref{apte:app:model}). We call $Q(\sigma)$ the decryption budget of a key; a key kept beyond it loses the guarantee, which does not mean that the key is recovered.
\paragraph{Stackelberg solution and feasibility.}
The defender also faces a control constraint: flooding the released command perturbs actuation, so $\sigma$ cannot exceed the largest level the coordination loop tolerates within its performance spec, denoted $\sigma_{\max}$.
Driving \cref{apte:eq:consensus} with a stochastic disturbance on the released command gives an $H_2$ norm that depends on the whole Laplacian spectrum rather than on $\rho_{\mathrm{cl}}$ alone, studied as network coherence for consensus~\cite{bamieh2012coherence,patterson2014consensus}.
With full actuation, and each agent's observer using the flooding it drew, the flooding reaches the formation through the modes of \cref{apte:eq:consensus} alone, so the map $H$ from flooding to centroid-removed formation error has, per axis and summed over the agents,
\begin{equation}
\label{apte:eq:hnorm}
\lVert H\rVert_2^2 = \sum_{j=2}^{M} \frac{1}{1-(1-\gamma-\varepsilon\mu_j)^2},
\end{equation}
a sum over every mode, while $\rho_{\mathrm{cl}}$ reads only the slowest.
With $\varrho$ the formation error the sensor noise alone produces, the stationary RMS formation error over the $2M$ position coordinates satisfies $\rho(\sigma)^2=\sigma^2\lVert H\rVert_2^2/M+\varrho^2$, so the specification $\rho_{\mathrm{spec}}$ leaves
\begin{equation}
\label{apte:eq:sigmax}
\sigma_{\max} = \frac{\sqrt{(\rho_{\mathrm{spec}}^2-\varrho^2)\,M}}{\lVert H\rVert_2}.
\end{equation}

\begin{proposition}[Passive commitment]
\label{apte:prop:passive}
Against a passive adversary with budget $\bar Q$, the Stackelberg-optimal flooding is the smallest $\sigma$ satisfying $Q(\sigma)\ge\bar Q$; this commitment keeps the key's guarantee over the mission if and only if this $\sigma$ does not exceed $\sigma_{\max}$, that is, if and only if $\bar Q\le Q_{\mathrm{key}}:=Q(\sigma_{\max})$.
\end{proposition}

The adversary's best response to any committed $\sigma$ is to spend its full budget, so the defender needs $Q(\sigma)\ge\bar Q$, and since the control cost increases in $\sigma$ it commits the smallest such $\sigma$, admissible exactly when it lies within $\sigma_{\max}$.

A persistent adversary's budget $\bar Q$ grows with the mission length, so no static commitment keeps the guarantee once $\bar Q$ exceeds $Q_{\mathrm{key}}$.
The defender's recourse is to reset the count of releases before it reaches $Q_{\mathrm{key}}$, which is precisely re-keying, and $Q_{\mathrm{key}}$ is directly the re-key horizon in cycles.

\section{Phase 2: The Active Phase and the Re-Keying Game}
\label{apte:sec:active}

Once the adversary turns active, the defender's levers are the flood level and its response to an alarm.

\subsection{Re-keying and its window}
\label{apte:subsec:guarantee}
\label{apte:subsec:equilibrium}

Re-keying resets the count of releases under the key to zero, which bootstrapping cannot do. With the key book of (A2) a re-key is a switch to the next key whose public and evaluation keys were uploaded in advance (\cref{apte:app:rekey}), so it interrupts nothing, and the period $T_{\min}$ of this pipeline bounds the key life from below. At the tolerable flood $\sigma_{\max}$ one key's guarantee covers
\begin{equation}
\label{apte:eq:qkey}
Q_{\mathrm{key}} = Q(\sigma_{\max}) = \frac{\sigma_{\max}^2}{c_f\,2^{\nu}B^2}
\end{equation}
cycles, the re-key horizon.

\begin{corollary}[Securability floor]
\label{apte:cor:floor}
Whenever $\sigma_{\max}<\sqrt{c_f2^{\nu}}B$, the right-hand side of \cref{apte:eq:qkey} drops below one, and no flooding within $\sigma_{\max}$ and no re-key cadence gives the $\nu$-bit guarantee.
\end{corollary}

At $\Delta=2^{59}$, $M=196$ and $\rho_{\mathrm{spec}}=0.05$~m the floor lies at $\nu=43$ to $47$ across the five graphs, and with pipelined keys the horizon must also hold one pipeline period, which binds first at $\nu=34$ to $40$ (\cref{apte:thm:cost}(i)).

Outside the contested stretch the defender chooses the flood and the key life to minimize the formation mean square that security adds; \cref{apte:thm:cost} in \cref{apte:app:cost} solves this in closed form, and with keys delivered ahead of use the optimum re-keys every $T_{\min}$ cycles at the cost of security
\begin{equation}
\label{apte:eq:jpipe}
J^\star = \frac{c_f\,2^{\nu}B^2\,\lVert H\rVert_2^2\,T_{\min}}{M},
\end{equation}
a factor $\bar Q/T_{\min}$ below one key for the whole mission, in which the graph enters three times, through $\lVert H\rVert_2$, $B$ and $T_{\min}$.

\subsection{The regimes of the active game}
\label{apte:subsec:regimes}

\begin{figure*}[t]
\centering
\includegraphics[width=\textwidth]{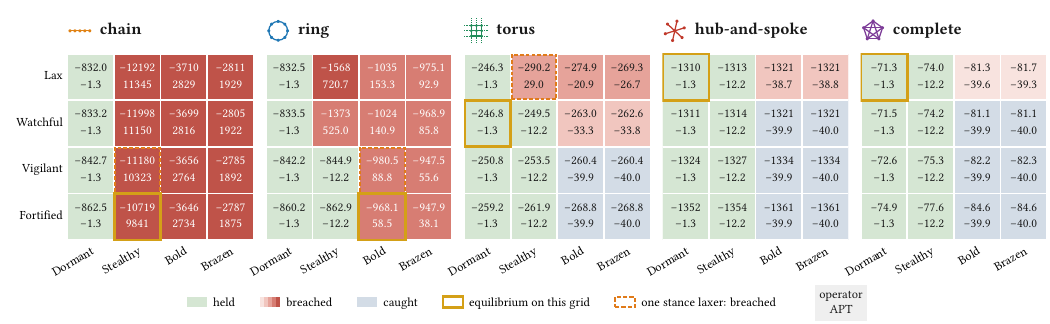}
\caption{Payoffs of the operator (top) and the APT (bottom) in each cell, in cm$^2$ of formation excess a day at $196$ agents, with a key priced at $10$~cm$^2$ and an alarm costing the APT $40$~cm$^2$.}
\label{apte:fig:payoff}
\end{figure*}

On segment $s$ let $\Delta_s=\rho_{\mathrm{spec},s}^2-\varrho^2$ be the error budget that the sensor noise leaves, and $m_s$ the margin that the flood at $\sigma_{\max,s}$ keeps. A flood $\sigma_s$ fills $(\Delta_s-m_s)\sigma_s^2/\sigma_{\max,s}^2$ and leaves the rest free for a spoof. A spoof of size $\lambda$ moves a window's mean square by at most $\lambda\Delta_s/R_s$, where $R_s$, the size of the spoof that uses the whole budget, is set by the graph: a weakly connected group hides a larger spoof behind the same residual.

A segment is \emph{deterred} when every schedule of spoof sizes that the APT weakly prefers to staying dormant leaves an expected excess of at most $\epsilon=10^{-3}$ in units of $\Delta_s/R_s$, an excess we count as none. Measuring the free part in spoof sizes, a larger free part deters more, because a spoof must then be larger to exceed the specification and the scan catches it sooner. Let $C_s$ be the least free part that deters segment $s$; it depends on the segment only through $c_{\mathrm{det}}R_s/\Delta_s$, the alarm cost in units of a unit spoof's excess, does not grow with it, and at zero alarm cost is certified at $5.8<C_s\le6.88$ (\cref{apte:app:model}). To be deterred, segment $s$ must leave free the share
\begin{equation}
\label{apte:eq:reach}
\kappa_s=\frac{C_s}{R_s}
\end{equation}
of its error budget.

\begin{proposition}[Regimes of the active game]
\label{apte:prop:regime}
(i) If $\kappa_s\le m_s/\Delta_s$, every flood up to $\sigma_{\max,s}$ deters segment $s$, and we set $\sigma_s^{\mathrm{d}}=\sigma_{\max,s}$. (ii) If $m_s/\Delta_s<\kappa_s<1$, the highest flood that deters segment $s$ is
\begin{equation}
\label{apte:eq:fdeter}
\sigma_s^{\mathrm{d}}=\sigma_{\max,s}\sqrt{\frac{1-\kappa_s}{1-m_s/\Delta_s}}<\sigma_{\max,s},
\end{equation}
which leaves the share $\kappa_s$ of the error budget free. (iii) If $\kappa_s\ge1$, no flood deters segment $s$: at every flood level the APT has a spoof that it weakly prefers to staying dormant and whose expected excess exceeds $\epsilon$.
\end{proposition}

\begin{proof}
Since a larger free part deters more, the flood $\sigma_s$ deters segment $s$ exactly when it leaves at least $\kappa_s\Delta_s$ free, that is when $(\Delta_s-m_s)\sigma_s^2/\sigma_{\max,s}^2\le(1-\kappa_s)\Delta_s$. For $\kappa_s<1$ this reads $\sigma_s\le\sigma_{\max,s}\sqrt{(1-\kappa_s)/(1-m_s/\Delta_s)}$, a bound that reaches $\sigma_{\max,s}$ exactly when $\kappa_s\le m_s/\Delta_s$. For $\kappa_s\ge1$ the free part stays below $\Delta_s\le\kappa_s\Delta_s$ for every flood $\sigma_s>0$.
\end{proof}

\begin{theorem}[Equilibrium over both phases]
\label{apte:thm:joint}
At a Stackelberg equilibrium $(\sigma^\star,\lambda^\star(\sigma^\star))$ of \cref{apte:eq:payoffA,apte:eq:payoffD}:
(i) every key keeps its budget, so in both phases the APT breaks any key with probability below $2^{-\nu}$;
(ii) on a segment with $\kappa_s<1$ the equilibrium flood is $\sigma_s^\star=\sigma_s^{\mathrm{d}}$ at every price $r$ of a key up to a level that grows with $c_{\mathrm{det}}$, and at every price when the deterring flood adds fewer keys than the alarm re-keys it prevents; on a segment with $\kappa_s\ge1$ the APT spoofs at every flood level, so $\sigma_s^\star$ trades the keys a lower flood consumes against the excess it removes;
(iii) a key serving segment $s$ lasts
\begin{equation}
\label{apte:eq:tstar}
T_s^\star = \bigl\lfloor Q(\sigma_s^\star)\bigr\rfloor
\end{equation}
cycles, so while it serves segment $s$ its budget is spent $\sigma_{\max,s}^2/\sigma_s^{\star2}$ times faster, and the rest of the route keeps the cost of security of \cref{apte:thm:cost}.
\end{theorem}

\begin{proof}
(i) holds because $K$ counts only keys that keep their budget. For (ii), write $\bar X(\sigma)$ and $\bar P(\sigma)$ for the excess and the alarm probability at the APT's best response, ties going to the defender, so that the defender's cost per unit $w_{\mathrm{ex}}$ is $r(K(\sigma)+\bar P(\sigma))+\bar X(\sigma)$. On a flood that deters segment $s$ the excess counts as none and the alarm probability is the dormant one, so among these floods the cost is lowest at the highest, $\sigma_s^{\mathrm{d}}$, because a lower flood consumes no fewer keys; at $r=0$ the tie goes to fewer keys. On a flood $\sigma>\sigma_s^{\mathrm{d}}$ the APT plays a spoof that it prefers to staying dormant, so $\bar X(\sigma)\ge c_{\mathrm{det}}\pi(\sigma)$, where $\pi(\sigma)$ is the alarm probability the spoof adds. The defender keeps $\sigma_s^{\mathrm{d}}$ when $r\bigl(K(\sigma_s^{\mathrm{d}})-K(\sigma)-\pi(\sigma)\bigr)\le\bar X(\sigma)$. With $D=K(\sigma_s^{\mathrm{d}})-K(\sigma_{\max,s})$ and $\underline{\pi}$ the least $\pi(\sigma)$ over $\sigma>\sigma_s^{\mathrm{d}}$, this holds at every $r$ when $D\le\underline{\pi}$ and for every $r\le c_{\mathrm{det}}\underline{\pi}/(D-\underline{\pi})$ otherwise. For $\kappa_s\ge1$ the APT's best response is a spoof by \cref{apte:prop:regime}(iii). (iii) follows from \cref{apte:eq:budget} at the flood $\sigma_s^\star$.
\end{proof}

The graph and the alarm cost place each segment in one of the three regimes. \Cref{apte:fig:payoff} samples the game at $196$ agents at four flood levels, the operator's stances, and four spoof sizes, the APT's postures (\cref{apte:app:profiles}). On the hub-and-spoke and the complete graph the scan and the alarm cost deter the APT at the tolerable flood ($\kappa_s$ of $0.003$ and $0.002$). On the torus a slightly lower flood deters it ($\kappa_s=0.05$): at Watchful the APT stays dormant, while one stance laxer a stealthy spoof breaches. On the ring and the chain ($\kappa_s$ of $2.2$ and $50$) every stance leaves a spoof that breaches, and the operator takes the stance that limits the damage. As the fleet grows $R_s$ falls at a rate the graph sets, so $\kappa_s$ grows: the ring enters the third regime at $196$ agents and the chain at $100$, while up to $576$ agents the torus stays in the second and the hub-and-spoke and the complete graph in the first.

\section{Experiments}
\label{apte:sec:evaluation}

The fleet has $M=196$ of the double integrators of \cref{apte:sec:setting}, with $\gamma=0.05$ and $dt=0.1$~s, each agent measuring its own position under sensor noise of covariance $0.01\,I$ and the relative position of each graph neighbor under noise of $0.03$~m per axis (A5).
The balanced consensus rate $\varepsilon^\star=2(1-\gamma)/(\mu_2+\mu_M)$~\cite{olfati2004consensus} gives every topology the same $\rho_{\mathrm{cl}}=1-\gamma$ by \cref{apte:eq:rhoconst}, the torus being the $14$ by $14$ grid.
The flooding fills $0.98$ of the error budget each specification leaves, every encrypted run is at \texttt{HEStd\_128\_classic} (\cref{apte:app:error}), and the security target is $\nu=30$ bits per key.

\paragraph{The mission.}
A convoy of three vehicles runs a closed lap of $3.6$~km at $2$~m/s, one lap every $30$ minutes and $48$ laps in a day. The fleet escorts it as a $14$ by $14$ lattice block whose spacing follows the clearance of each segment, which gives specifications of $0.05$~m where the lap is confined, $0.16$~m in open water and $0.09$~m in between (\cref{apte:app:mission}). In the east turn the block's two outer rows, a detachment of $28$ agents, leave on a detour of $15$ minutes and merge back in the west turn. In the contested version the detour is the contested zone, and the APT spoofs the detachment's position sensors on its return leg.

\paragraph{Control-tolerable flooding and topology.}
All five topologies share $\rho_{\mathrm{cl}}=0.95$, yet the flooding they tolerate differs by a factor of $7.7$ in variance, from $\sigma_{\max}=0.0149$~m on the hub-and-spoke to $0.0414$~m on the complete graph, and the horizon by a factor of $18$, from $7.8\times10^{3}$ to $1.4\times10^{5}$ cycles (\cref{apte:tab:fleet}). The squared $H_2$ norm per agent spans $0.995$ to $8.60$ across the same five graphs and determines $\sigma_{\max}$ outright rather than predicting it, so the security horizon is denominated in the sum over modes that the radius cannot see. Every graph is inside the window at $\nu=30$, so re-keying is the operative defense on all five (\cref{apte:app:mission}). At the tight specification a key lasts $56$ minutes on the torus and $14$ minutes on the ring, and the route-aware law re-keys $26$ and $95$ times in a day, against $34$ and $104$ for a cadence blind to the route (\cref{apte:app:mission}).

\paragraph{The equilibrium.}
At an alarm cost of $40$~cm$^2$ and a key price of $10$~cm$^2$ (\cref{apte:fig:payoff}), the torus's outer segments must leave free $5\%$ of the error budget, so the equilibrium flood is $0.986$ of the tolerable one, and because it adds fewer keys than the alarm re-keys it prevents, it is the equilibrium at every key price (\cref{apte:thm:joint}). Deterring the APT costs the torus $24.64$ keys a day against $24.60$ at the tolerable flood. On the hub-and-spoke and the complete graph the alarm cost alone deters the APT at the tolerable flood whenever it exceeds $2.7$ and $1.6$~cm$^2$. On the ring every stance leaves a breach, and the operator's best stance spends $2.8$ more keys a day than the laxest to cut the APT's excess $7.4$ times, from $733$ to $98$~cm$^2$ a day; on the chain the same keys cut it only $1.15$ times.
\paragraph{The error model at 128-bit security.}
Ten honest cycles on the torus at \texttt{HEStd\_128\_classic} give a mean square entry error $0.98$ times the model's behind \cref{apte:lem:honest}, a largest entry error of $0.07B$, and a neighbor sum of $94$~ms within the $100$~ms period (\cref{apte:app:error}).
\section{Conclusion}
\label{apte:sec:conclusion}

We modeled encrypted multi-agent control under an advanced persistent threat as one leader-follower game in which the adversary collects decryptions against the key and spoofs the agents' sensors. The key's budget is set by the loop's tolerance to flooding, so how the fleet is wired sets how soon the key must be replaced at identical closed-loop stability, and in the active phase the coordination graph and the APT's cost of an alarm decide in closed form whether the scan alone, a lower flood or nothing deters the spoofer. The claims are bounded by full observation, a single cooperative-fleet defender, an honest-but-curious coordinator and a single re-key at the first alarm; strategic and corrupted agents belong to the per-agent-key setting, where the single shared key this paper turns into a game is itself distributed.

\bibliographystyle{ACM-Reference-Format}
\bibliography{references}

\clearpage
\appendix
\crefalias{section}{appendix}
\section{Configuration and Honest Error}
\label{apte:app:error}

\paragraph{Configuration.}
The lattice security of the ciphertexts is fixed by $N$ and $q$ through the Homomorphic Encryption Standard~\cite{HE-standard}.
Every encrypted run uses OpenFHE~1.4.2 through OpenFHE.jl~\cite{schlottkelakemper2024openfhejulia} at its 128-bit level \texttt{HEStd\_128\_classic}, with ring dimension $16384$, $8192$ entries per ciphertext, depth $2$ and scale $2^{59}$, where (T3) leaves $c=128-\log_2 24\approx123.4$ bits.
The regimes and the cost of security use $\Delta=2^{59}$, which needs the same ring dimension as $2^{50}$ on every graph, so precision costs nothing here. With one copy of the fleet vector every graph's layout fits the $8192$ entries, the largest rotation being $392$ entries at $M=196$, so the ring dimension is set by the security level and not by the layout, and the library accepts $16384$ and refuses $8192$ at this level.
For a target of $\nu$ bits the failure probabilities of \cref{apte:lem:honest} are set at $2^{-(\nu+3)}/\bar Q$ per cycle and $2^{-(\nu+3)}$ per key, so that every failure of the guarantee stays below $2^{-\nu}$ per key.
The fleet's formation offsets lie on a grid of side $2.0$~m.
Each agent's observer is a steady-state Riccati design with process weight $\sigma_{\max}^2AA^\top$, a fixed point in $\sigma_{\max}$ that we solve by bisection, and the closed form \cref{apte:eq:hnorm} and the sensor-driven error $\varrho^2$ match the stationary covariance of the full $196$-state loop per axis to a relative $10^{-8}$ on every graph.
In the experiments the flooding fills $0.98$ of the error budget so that the specification holds on every run and not only in expectation, and every agent runs one fixed observer gain, at which the estimation error decays by $0.95$ per cycle, the formation's own rate.
The agents draw the flooding themselves and add it to the released target, rather than through the library's flooding mode, so that each agent's observer knows the flood it drew.
In the encrypted run of \cref{apte:sec:evaluation} a cycle takes $0.99$~s, of which the $196$ encryptions take $0.82$~s at $4.2$~ms each and run at the agents in parallel, the coordinator's neighbor sum $94$~ms against the control period of $100$~ms, and one decryption $20$~ms.
The other graphs' numbers follow from the spectrum and the error bound of \cref{apte:lem:honest}, which this run checks on the torus.

\paragraph{Packing.}
Each agent's block of the packed ciphertext holds its $n$ components, and the coordinator rotates by multiples of $n$ entries to reach another agent, so it applies the same plaintext weights to every agent's block.
A packed ciphertext uses $nM$ entries, $784$ of $8192$ at $M=196$ and $n=4$, plus a window of $2n$ entries for the shifts of the ring graph.
On the ring graph the coordinator's circuit is the sum of the $196$ uploads, two plaintext products and two rotations, both rotations acting on degree-two ciphertexts, at multiplicative depth one.

\paragraph{Honest error.}
The budget of \cref{apte:eq:budget} rests on a bound for the honest error of every released entry, derived from the circuit.

\begin{lemma}[Honest error]
\label{apte:lem:honest}
Let the coordinator compute the neighbor sum with weights $W$ in OpenFHE~1.4.2 CKKS, with a uniform ternary key, encryption error of standard deviation $\sigma_e=3.19$, $b=N/2$ entries per ciphertext and uploads clipped to $[-w,w]$.
On a key event of probability at least $1-\delta_{\mathrm{key}}$ and, in each cycle, outside an event of probability at most $\alpha_c$, every released entry's honest error is at most
\begin{equation*}
B = z_d\sqrt{P} + R + z_g\sqrt{V}.
\end{equation*}
Here $P=b\,w_2\lambda_{\max}/\Delta^2$ bounds the variance proxy of the fresh encryption error carried through the weights, with $w_2=\max_i\sum_j W_{ij}^2$ and $\lambda_{\max}=M\sigma_e^2(4Nu/3+1)$ the key event's bound on the key's root values, with $u$ a quantile set by $\delta_{\mathrm{key}}$. The term $R=\lVert W\rVert_\infty MN/(2\Delta)+p\,Nw/(2\Delta)$ bounds the encoding rounding of the uploads and of the $p$ plaintext products at its worst case. The term $V$ bounds the variance of the decoder's own Gaussian through a chi-square tail on the decrypted imaginary parts. The quantiles $z_d$ and $z_g$ come from Gaussian union bounds over the $Mn$ released entries of a cycle.
\end{lemma}

\begin{proof}[Sketch]
Five sources make up the honest error.
The fresh encryption of the $M$ uploads is linear in independent coordinates with sub-Gaussian proxies given the key, and distinct entries are orthogonal functionals of the coefficients.
The key event bounds the key's values at the roots by a union bound over the $N/2$ conjugate pairs.
The encoding rounding of the uploads and of the plaintext weights is taken at its worst case.
The key switches of the two rotations act on degree-two ciphertexts and are divided by a second scale factor, which leaves them negligible.
The decoder adds a Gaussian whose variance the library sets from the decrypted imaginary parts, and a chi-square tail bounds that variance.
A union bound over the released entries gives $B$.
\end{proof}

The bound uses no rounding heuristic, and the honest cycles of \cref{apte:sec:evaluation} on the torus at $\Delta=2^{59}$ reproduce its error model.

\section{Re-Key Delivery}
\label{apte:app:rekey}

With the key book a re-key delivers no secret key; the coordinator needs the next epoch's public key and evaluation keys, which the designated agent of (A6) generates from the book and uploads.
The coordinator needs rotation keys for the offsets between agents that the graph uses in the packing, at $M=196$: $2$ for the chain and the ring, $6$ for the torus and $195$ for the hub-and-spoke and the complete graph.
Generating the keys at \texttt{HEStd\_128\_classic} and ring dimension $16384$ takes at most $0.019$~s with two rotation keys, $0.038$~s with six and $0.965$~s with $195$, upper bounds from timings at a larger depth than the circuit's.
The evaluation keys take $6.0$, $16.5$ and $512$~MB, derived from the key structure, and over an assumed $100$~Mbit/s link to the coordinator they take $0.48$, $1.32$ and $40.96$~s.
The period of the key pipeline is its slowest stage counted in cycles of $dt$, which gives $T_{\min}=5$, $14$ and $410$.
A cold re-key stops the loop until those keys are generated and uploaded, so its outage lasts $5$, $14$ and $420$ cycles.

\section{Stances and Postures}
\label{apte:app:profiles}

The operator's four stances set the flood on the contested segments to a share of the tolerable flood, and the APT's four postures hold one spoof size in every attacked window of the day. \Cref{apte:tab:stances} gives the keys a day each stance consumes, and \cref{apte:tab:postures} gives each posture's chance of an alarm. The alarm cost of $40$~cm$^2$ is the geometric middle of the range, $14.4$ to $100$~cm$^2$, over which every equilibrium of \cref{apte:fig:payoff} holds against every spoof size; the torus keeps its equilibrium at every key price, and the ring keeps its own up to $16$~cm$^2$ per key.

\begin{table}[h]
\centering
\caption{Stances: flood over the tolerable flood and keys a day at $196$ agents.}
\label{apte:tab:stances}
\small
\begin{tabular}{lccccc}
\toprule
Stance & Flood & Chain & Ring & Torus & Hub / complete \\
\midrule
Lax & $1.00$ & $83.2$ & $83.2$ & $24.6$ & $131.0$ / $7.10$ \\
Watchful & $0.95$ & $83.3$ & $83.3$ & $24.6$ & $131.1$ / $7.11$ \\
Vigilant & $0.70$ & $84.2$ & $84.2$ & $25.0$ & $132.4$ / $7.23$ \\
Fortified & $0.50$ & $86.2$ & $86.0$ & $25.9$ & $135.1$ / $7.46$ \\
\bottomrule
\end{tabular}
\end{table}

\begin{table}[h]
\centering
\caption{Postures: spoof size, chance of an alarm in the day and in the first window, and the largest bias on one agent of the torus.}
\label{apte:tab:postures}
\small
\begin{tabular}{lcccc}
\toprule
Posture & Size & Alarm in the day & First window & Largest bias \\
\midrule
Dormant & $0$ & $0.03$ & $0.001$ & $0$ \\
Stealthy & $0.5$ & $0.30$ to $0.31$ & $0.011$ & $3.6$ to $4.1$~cm \\
Bold & $2$ & $0.997$ & $0.56$ to $0.61$ & $7.3$ to $8.2$~cm \\
Brazen & $6$ & $1.00$ & $1.00$ & $12.6$ to $14.2$~cm \\
\bottomrule
\end{tabular}
\end{table}

\section{Model Details}
\label{apte:app:model}

\paragraph{Flooding.}
The budget \cref{apte:eq:budget} holds provided CKKS is $(c+\log_2 24)$-bit IND-CPA secure (T3) and each agent draws its flooding from a discrete Gaussian on a fixed grid, whose sampling costs the one bit counted in $c_f$. Reading the advantage as the probability that a key's releases are compromised, $2^{-\nu}$ bounds the expected fraction of the mission's traffic exposed, for every cadence. A target per mission, that no key is ever broken, would need $\sigma^2\ge c_f\,2^{\nu}B^2\bar Q$ whatever the cadence and so remove the statistical benefit of re-keying, a point Li et al.\ raise for protocols instantiated many times~\cite[p.~21]{li2022securing}.

\paragraph{The least deterring free part.}
We compute $C_s$ by dynamic programming over every schedule of spoof sizes across the $36$ attacked windows of a day, the campaign ending at an alarm and the alarm cost charged once. At zero alarm cost it equals $5.934$, inside the certified bracket $(5.8,6.88]$; at the alarm cost of \cref{apte:fig:payoff} it is $4.49$ on the chain, $2.87$ on the ring, $0.588$ on the torus, $0.199$ on the hub-and-spoke and $0.154$ on the complete graph. The excess and the alarm probability depend on the graph only through the free part and the unit $\Delta_s/R_s$, so one function of the alarm cost serves all five graphs at a given fleet size.

\section{The Cost of Security}
\label{apte:app:cost}

Three constraints bind at once: a larger security target $\nu$ shrinks the budget at every flooding level, $\sigma_{\max}$ caps the noise the loop tolerates, and each re-key perturbs the loop, by a formation cost $c_{\mathrm{rk}}=c_0+c_1\sigma^2$ in m$^2$ cycles when it stops the loop, or by bounding the key life below by $T_{\min}$ when keys are delivered ahead of use.
In the collection phase the adversary only records, so on every segment outside the contested stretch the defender chooses $\sigma$ and $T$ together to minimize the formation mean square that security adds over a mission of $\bar Q$ cycles,
\begin{equation}
\label{apte:eq:leader}
J(\sigma,T) = \frac{\sigma^2\lVert H\rVert_2^2}{M} + \frac{\lceil\bar Q/T\rceil-1}{\bar Q}\,c_{\mathrm{rk}},
\end{equation}
over $\sigma\in(0,\sigma_{\max}]$ and integer $T\in[T_{\min},\bar Q]$ with $T\le Q(\sigma)$, the first key being made before the mission.

\begin{theorem}[Cost of security]
\label{apte:thm:cost}
Let the budget \cref{apte:eq:budget} hold except on an event of probability at most $2^{-\nu}$ per key, and write $\beta=c_f\,2^{\nu}B^2$, $a=\lVert H\rVert_2^2/M$ and $T_{\mathrm{cap}}=\lfloor Q_{\mathrm{key}}\rfloor$.
(i) Problem \cref{apte:eq:leader} is feasible if and only if $T_{\mathrm{cap}}\ge T_{\min}$.
(ii) For each admissible $T$ the optimal flooding is $\sigma(T)=\sqrt{\beta T}$, the least flooding whose budget \cref{apte:eq:budget} covers $T$ cycles.
(iii) With keys delivered ahead of use, $c_{\mathrm{rk}}=0$, the optimum re-keys every $T_{\min}$ cycles at the cost \cref{apte:eq:jpipe}, a factor $\bar Q/T_{\min}$ below the cost $a\beta\bar Q$ of covering the mission with one key.
\end{theorem}

With $T_{\mathrm{eq}}=\sqrt{c_0/(a\beta)}$, re-keys that stop the loop put the optimum of \cref{apte:eq:leader} in one of five regimes.

\begin{proposition}[Regimes of the cost of security]
\label{apte:prop:costregimes}
The optimum is infeasible if $T_{\mathrm{cap}}<T_{\min}$; key rate, $T^\star=T_{\min}$, if $T_{\mathrm{eq}}\le T_{\min}$; interior, $T^\star\approx T_{\mathrm{eq}}$ and $J^\star\approx2\sqrt{ac_0\beta}+c_1\beta$, if $T_{\min}<T_{\mathrm{eq}}<T_{\mathrm{cap}}$ and $\bar Q\gg T_{\mathrm{eq}}$; at the specification, $T^\star=T_{\mathrm{cap}}$ and $\sigma^\star=\sigma_{\max}$, if $T_{\mathrm{cap}}\le T_{\mathrm{eq}}$ and $T_{\mathrm{cap}}<\bar Q$; and static, $T^\star=\bar Q$ with $J^\star=a\beta\bar Q$, optimal whenever $\bar Q\le\min(\sqrt2\,T_{\mathrm{eq}},T_{\mathrm{cap}})$.
\end{proposition}

\begin{corollary}[Re-keys that stop the loop]
\label{apte:cor:cold}
If each re-key stops the loop at a fixed cost $c_0>0$ with $c_1=0$, as when the agents hold their last released target through a re-key without a spare key, then in the interior regime $J^\star\approx2\sqrt{c_f c_0}\,2^{\nu/2}B\,\lVert H\rVert_2/\sqrt M$, so each added bit of per-key security multiplies the cost of security by $\sqrt2$, against $2$ in \cref{apte:eq:jpipe}.
\end{corollary}

Outside the static regime the optimum re-keys on the laziest cadence that the budget \cref{apte:eq:budget} of its own flooding allows, which under pipelined keys is the fastest the key pipeline permits.
Below the floor no defense within spec keeps the guarantee, between the floor and the ceiling where $Q_{\mathrm{key}}$ reaches $\bar Q$ re-keying is the operative defense, and above it static flooding suffices at $\bar Q/T_{\min}$ times \cref{apte:eq:jpipe}.
The efficient secure operating point therefore lies between a securability floor and a static-suffices ceiling.

\section{Mission and Contested Run}
\label{apte:app:mission}

\begin{table*}[h]
\centering
\small
\caption{The five topologies at $M=196$, each at its own balanced consensus gain, at $\nu=30$, $\Delta=2^{59}$ and the mission's tight specification $\rho_{\mathrm{spec}}=0.05$~m, with one copy of the fleet vector in ring dimension $16384$; all five have $\rho_{\mathrm{cl}}=0.95$.
$\varrho$ is the formation error the sensor noise alone produces, $Q_{\mathrm{key}}$ the re-key horizon \cref{apte:eq:qkey} at that specification, $T_{\min}$ the period of the key pipeline and $J^\star$ the cost of security \cref{apte:eq:jpipe} with pipelined keys.}
\label{apte:tab:fleet}
\begin{tabular}{@{}lccccccccc@{}}
\toprule
topology & $\varepsilon^\star$ & $\lVert H\rVert_2^2/M$ & $\varrho$ (m) & $\sigma_{\max}$ (m) & $B$ ($10^{-11}$~m) & $Q_{\mathrm{key}}$ (cycles) & rotation keys & $T_{\min}$ (cycles) & $J^\star$ (m$^2$) \\
\midrule
hub-and-spoke & 0.0096 & 8.60 & 0.023 & 0.0149 & 3.77 & $7.8\times10^{3}$ & 195 & 410 & $1.0\times10^{-4}$ \\
chain & 0.4750 & 3.15 & 0.029 & 0.0227 & 5.51 & $8.4\times10^{3}$ & 2 & 5 & $9.7\times10^{-7}$ \\
ring & 0.4749 & 3.14 & 0.029 & 0.0227 & 5.51 & $8.4\times10^{3}$ & 2 & 5 & $9.6\times10^{-7}$ \\
torus & 0.2318 & 1.51 & 0.028 & 0.0335 & 4.06 & $3.4\times10^{4}$ & 6 & 14 & $7.0\times10^{-7}$ \\
complete & 0.0048 & 0.995 & 0.027 & 0.0414 & 2.49 & $1.4\times10^{5}$ & 195 & 410 & $5.1\times10^{-6}$ \\
\bottomrule
\end{tabular}
\end{table*}

\paragraph{The mission.}
The lap has a bend around an obstacle on its southern leg. The block's spacing is $1.2$~m where the lap is confined, in its turns and at the bend, $2.7$~m in open water, and $1.8$~m in the transitions between them.
The specification of a segment follows from its spacing, $\rho_{\mathrm{spec}}=c/7$ with $c$ the half spacing less the agent's radius, which keeps the expected number of neighbor contacts below one per thousand days.
Each agent holds its own place in the block along the lap ahead of time, so the route itself costs no formation error (\cref{apte:sec:setting}).
The detachment investigates a contact behind a second obstacle north of the lap, loiters there and returns around the obstacle's west side.
While split, the main body coordinates on the induced sub-graph of the base topology and the detachment on its own, the cylinder of rows $1$ to $12$ and the ring of rows $13$ and $14$ on the torus; on the ring graph the detached arc is closed into a ring of its own, at the cost of two rotation keys, because its induced chain cannot defend itself (below).
The lap has nine segments, five confined or transitional at the specifications above, the three of the detour at $0.09$, $0.16$ and $0.09$~m for the detachment while the main body stays confined, and the merge.
In the contested version the adversary spoofs the detachment's return leg either with a spoof that the single-window test of the scan's longest block catches half the time or with its equilibrium spoof (\cref{apte:thm:joint}).

\paragraph{Key life along the route.}
The route sets the wear.
A confined segment wears the key $15$ times faster than open water on the ring and $17$ times faster on the torus, since the horizon scales with the error budget $\rho_{\mathrm{spec}}^2-\varrho^2$ that each specification leaves to flooding.
The complete graph re-keys $7$ times in a day and the hub-and-spoke $191$.
In the honest run the RMS formation error averaged over each segment lies within its specification on every graph and every seed, at most $0.99$ of it, and at the split, the merge and every change of specification the expected error stays below $0.9$ of twice the tighter neighboring specification.

\paragraph{The window.}
The floor of \cref{apte:cor:floor} lies at $\nu=43$ for the chain, the ring and the hub-and-spoke, $45$ for the torus and $47$ for the complete graph, so a fleet serving every topology has its security target below $43$ bits, and the pipeline binds first at $\nu=34$ on the hub-and-spoke, $38$ on the complete graph and $40$ on the other three. At the tight specification every graph is inside the window at $\nu=30$, so re-keying is the operative defense on all five, at a cost of security from $7.0\times10^{-7}$~m$^2$ on the torus to $1.0\times10^{-4}$ on the hub-and-spoke. At a specification of $0.5$~m the same fleet re-keys only above $\nu=30$ to $35$ and reaches its floor at $50$ to $54$, so the mission's tightest segment, not the fleet, sets where the game is played.

\paragraph{Fleet size and topology.}
Two laws govern how the mission scales.
The re-key horizon falls with the square of the fleet size, as $M^{-2.2}$ on the torus and $M^{-2.1}$ on the ring between $49$ and $400$ agents at the same spacing, because $c_f$ grows with the released entries, $B^2$ with the summed uploads, and $\sigma_{\max}^2$ with $M/\lVert H\rVert_2^2$, which is a graph constant per agent.
A fleet of $49$ agents on this lap uses less than a key a day; one of $196$ re-keys every hour on the torus.
The time that a spoof caught half the time by a single window's test needs to use up the error budget grows with the fleet too, but at a rate the graph sets.
On the torus it grows as $M^{1.9}$, and the detour's segments hold their specification against such a spoof up to $400$ agents.
On the ring the detached arc's connectivity falls as $M^{-2}$, so that time grows as $M^{5.5}$ for the closed ring and $M^{7}$ for the induced chain, and the detour stops holding past $196$ agents in the first case and $100$ in the second.
The torus therefore tolerates more flooding than the ring and defends the detour to a larger fleet, and the fleet size widens the gap in defense.

\paragraph{The contested run.}
\label{apte:app:contested}
In the detour's segments a spoof that a single window's test catches half the time uses at most $8\%$ of the error budget on the torus and $76\%$ on the ring, so flooding less to make room for it shortens the key by $1.09$ and $4.1$ times, and it needs $21$ cycles on the torus and $1718$ on the ring to use up the error budget, against segments of $3000$. Both graphs can therefore hold the specification against it at this relative-sensing noise; at $0.1$~m the torus still holds and the ring does not.
In closed loop that spoof sets off a single window's test in about half the attacked windows, as its size is defined, and the lower flood of the attacked segment holds the detachment's error at $0.085$~m against its $0.09$~m specification.
The merge is where the split is paid for.
A rigid shift of the whole detachment is invisible to the relative sensing inside it, and the largest stealthy one, $0.34$~m on the torus and $0.85$~m on the ring, enters the whole fleet's error at the merge with a peak of $2.8$ and $18$ times the specification's square.
In a rendezvous test over the last minute of the detour, the agents on either side of the cut measure their relative positions across it, and the designated agent runs the same energy test on the resulting residuals.
The test caps the shift at $0.027$ and $0.053$~m and the peak at $0.018$ and $0.068$, a reduction of $157$ and $266$ times.
With one re-key at the first alarm the fleet spends $27$ to $30$ keys a day on the torus and $97$ on the ring, against $49$ and $129$ if every alarm re-keyed, and the false-alarm level of $10^{-3}$ per window raises no false alarm in a day.
A spoof that a single window's test catches half the time crosses the scan's threshold in a few percent of its windows, so the alarm came $2$ to $14$ attacked laps after the attack began, during which the zone ran on its honest budget.
Pre-empting with a fresh key at the start of the return leg cuts the keys consumed under attack from $43$ to $40$ on the torus and from $72$ to $42$ on the ring.

\section{Proofs}
\label{apte:app:proofs}

\begin{proof}[Proof of \cref{apte:thm:cost}, \cref{apte:prop:costregimes} and \cref{apte:cor:cold}]
For (i) and (ii), $J$ increases in $\sigma$ at fixed $T$, and $T\le Q(\sigma)$ is $\sigma^2\ge\beta T$, so the least such $\sigma$ is optimal, and it lies within $\sigma_{\max}$ exactly when $T\le T_{\mathrm{cap}}$.
For \cref{apte:prop:costregimes}: with $\sigma(T)^2=\beta T$, $J$ becomes $j(T)=a\beta T+\bigl((\lceil\bar Q/T\rceil-1)/\bar Q\bigr)(c_0+c_1\beta T)$, which on each set of $T$ with the same $\lceil\bar Q/T\rceil=k+1$ has slope $a\beta+kc_1\beta/\bar Q>0$, so $T^\star$ is the least element of one such set, clipped to $[T_{\min},\min(\bar Q,T_{\mathrm{cap}})]$.
For $\bar Q\gg T$, replacing $(\lceil\bar Q/T\rceil-1)/\bar Q$ by $1/T$ gives $a\beta T+c_0/T+c_1\beta$, which is least at $T_{\mathrm{eq}}$ with value $2\sqrt{ac_0\beta}+c_1\beta$, and the clipping at $T_{\min}$ and $T_{\mathrm{cap}}$ gives the key-rate and specification regimes.
The static cadence $T=\bar Q$ costs $a\beta\bar Q$, and against $k\ge1$ re-keys, whose least cost is at least $a\beta\bar Q/(k+1)+kc_0/\bar Q$, it is no more costly when $\bar Q^2\le(k+1)T_{\mathrm{eq}}^2$, which binds at $k=1$.
For (iii) of \cref{apte:thm:cost}: with $c_0=c_1=0$, $j(T)=a\beta T$ increases in $T$, so $T^\star=T_{\min}$ and $J^\star=a\beta T_{\min}$, which is \cref{apte:eq:jpipe}.
For \cref{apte:cor:cold}, substituting $a$ and $\beta$ into $2\sqrt{ac_0\beta}$ gives the stated value, which grows as $2^{\nu/2}$, while \cref{apte:eq:jpipe} grows as $2^{\nu}$.
\end{proof}

\end{document}